\documentclass[amsmath,amssymb]{revtex4}
\usepackage{graphicx}
\usepackage{amscd,amsmath,amsthm,amsfonts,amssymb}

\newtheorem{theorem}{Theorem}

\newtheorem{remark}{Remark}

\begin{document}
\title{On transverse stability of discrete line solitons}
\author{Dmitry E. Pelinovsky$^{1}$ and Jianke Yang$^{2}$}

\address{
$^1$ Department of Mathematics and Statistics, McMaster University,
Hamilton, Ontario, Canada, L8S 4K1 \\
$^2$ Department of Mathematics and Statistics, University of
Vermont, Burlington, VT 05401, USA }

\begin{abstract}
We obtain sharp criteria for transverse stability and instability of
line solitons in the discrete nonlinear Schr\"{o}dinger equations on
one- and two-dimensional lattices near the anti-continuum limit. On
a two-dimensional lattice, the fundamental line soliton is proved to
be transversely stable (unstable) when it bifurcates from the $X$
($\Gamma$) point of the dispersion surface. On a one-dimensional
(stripe) lattice, the fundamental line soliton is proved to be
transversely unstable for both signs of transverse dispersion. If
this transverse dispersion has the opposite sign to the discrete
dispersion, the instability is caused by a resonance between
isolated eigenvalues of negative energy and the continuous spectrum
of positive energy. These results hold for both focusing and
defocusing nonlinearities via a staggering transformation. When the
line soliton is transversely unstable, asymptotic expressions for
unstable eigenvalues are also derived. These analytical results are
compared with numerical results, and perfect agreement is obtained.
\end{abstract}

\maketitle

\section{Introduction}

One-dimensional solitons, when viewed in two spatial dimensions,
become line solitons which are uniform along the line direction
(called transverse direction). Thus an important physical question
is the transverse stability of line solitons to transverse
perturbations. It is well known that in homogeneous media, line
solitons in the nonlinear Schr\"odinger (NLS) equation and other
related wave equations are always transversely unstable
\cite{ZakRub} (see also \cite{TI_saturable,SHG_TI_neck,SHG_TI_snake}
for applications in optics and \cite{KivPel,Yang_SIAM} for reviews).
This instability has been observed in recent optical experiments
\cite{Gorza,Gorza2,Mamaev}. In the presence of a one-dimensional
periodic potential, many line solitons are still transversely
unstable \cite{Aceves_semi_TI,1Dlattice_TI}. To suppress this
transverse instability, various techniques have been proposed
\cite{Anasta,Mussli,MussYang,Yang1,Yang2}. In particular, it was
shown numerically in \cite{Yang1,Yang2} that when a one- or
two-dimensional periodic potential is included in the continuous NLS
equation, this transverse instability can be completely eliminated
if the line soliton bifurcates from certain symmetry points of the
dispersion surface. But in the corresponding discrete nonlinear
Schr\"{o}dinger (dNLS) equations, discrete line solitons in
one-dimensional lattices are numerically found to be still
transversely unstable \cite{Yang2}, highlighting the difference
between continuous and discrete NLS models.

In this article, we analytically investigate transverse stability of
line solitons in the dNLS equations on two-dimensional (square) and
one-dimensional (stripe) lattices. Near the anti-continuum limit, we
derive sharp stability criteria for these discrete line solitons. We
prove for a two-dimensional lattice that the fundamental line
soliton is transversely stable (unstable) when it bifurcates from
the $X$ ($\Gamma$) point of the dispersion surface. For a
one-dimensional (stripe) lattice, the fundamental line soliton is
proven to be transversely unstable for both signs of transverse
dispersion. These results hold for both focusing and defocusing dNLS
equations via a staggering transformation. For unstable line
solitons, their unstable eigenvalues are also derived
asymptotically. We also investigate the transverse stability of line
solitons numerically both near the anti-continuum limit and away
from it. Near the anti-continuum limit, the numerical results fully
agree with the analytical results. Away from the anti-continuum
limit, we reveal additional bifurcations of unstable eigenvalues
which cannot be captured by the theoretical analysis.

This article is organized as follows. In Section \ref{sec:2D}, we
consider the transverse stability of discrete line solitons on a
two-dimensional lattice. We show that the entire solution family
bifurcating from the $\Gamma$ point is transversely unstable, whereas
the solution family bifurcating from the $X$
point is transversely stable in the anti-continuum limit. In Section
\ref{sec:1D}, we consider the transverse stability of discrete line
solitons on a one-dimensional lattice, and show that they are always
unstable. Numerical results and their comparison with the theory are
reported in Section \ref{sec:num}. Section \ref{sec:summary}
concludes the paper with discussion of open problems.

Before we start, we first introduce some mathematical notations
which will be used in later analysis. If $\{ \psi_n\}_{n \in
\mathbb{Z}}$ is a bi-infinite sequence (i.e., a sequence which is
infinite in both directions), and $\mathbb{Z}$ is the set of
integers, then $\psi$ denotes the vector for this sequence in some
vector space such as $l^2(\mathbb{Z})$ or, more generally,
$l^p(\mathbb{Z})$ for $p
\geq 1$. Here $l^2(\mathbb{Z})$ denotes the space of bi-infinite
squared-summable sequences with the norm $\| \psi \|_{l^2} \equiv
\left( \sum_{n \in \mathbb{Z}} |\psi_n|^2 \right)^{1/2}$ and the
inner product $\langle \psi, \varphi \rangle \equiv \sum_{n \in
\mathbb{Z}} \bar{\psi}_n \varphi_n$, with the overbar for complex
conjugation, and $l^p(\mathbb{Z})$ denotes the space of sequences
with the norm $\| \psi \|_{l^p} \equiv \left( \sum_{n \in
\mathbb{Z}} |\psi_n|^p \right)^{1/p}$.

\section{Two-dimensional lattice} \label{sec:2D}

In this section, we study transverse stability of line
solitons in the dNLS equation on a two-dimensional lattice,
\begin{equation} \label{dNLS}
i \frac{d u_{m,n}}{d t} + \epsilon (u_{m+1,n} + u_{m-1,n} + u_{m,n+1} + u_{m,n-1} - 4 u_{m,n}) + |u_{m,n}|^2
u_{m,n} = 0,
\end{equation}
where $(m,n) \in \mathbb{Z}^2$, $u_{m,n}$ are complex-valued
amplitudes that depend on the evolution time $t$,  and $\epsilon$ is
the lattice-coupling constant. Here the sign of nonlinearity has
been normalized to be unity through a scaling of $\epsilon$ and $t$.
The anti-continuum limit $\epsilon = 0$ of zero coupling between
lattice sites was found to be very attractive for many analytical
studies on the existence and stability of discrete solitons in the
framework of the dNLS equation \cite{Chong,MA94,PKF1,PelSak}.
Detailed account of mathematical results obtained in the
anti-continuum limit can be found in the monograph \cite{Pel-book}.

In the above dNLS equation, the defocusing case $\epsilon < 0$ can
be mapped to the focusing case $\epsilon > 0$ by the staggering
transformation
\begin{equation}
u_{m,n}(t) = (-1)^{m+n} v_{m,n}(t) e^{-8 i \epsilon t}.
\end{equation}
If $u$ solves the dNLS equation (\ref{dNLS}), then $v$ solves the
same equation with $\epsilon$ replaced by $-\epsilon$. Thus, in what
follows, we will consider the focusing case ($\epsilon>0$) only.

The linear dispersion surface of the dNLS equation (\ref{dNLS}) is
given by the function
\begin{equation*}
\omega(k,p) = \epsilon (4 - 2 \cos(k) - 2 \cos(p)) = 4 \epsilon \left[ \sin^2\left(\frac{k}{2}\right)
+ \sin^2\left(\frac{p}{2}\right) \right],
\end{equation*}
where wavenumbers $(k,p)$ reside in the first Brillouin zone
$[-\pi,\pi] \times [-\pi,\pi]$. This dispersion relation can be
derived by substituting the discrete Fourier modes $u_{m,n}(t) =
e^{i k m + i p n - i \omega t}$ into the linear dNLS equation
(\ref{dNLS}).

To understand bifurcations of stationary line solitons in Eq.
(\ref{dNLS}), we need to classify the stationary points of the
dispersion surface, where $\nabla \omega(k,p) = 0$. In the semi-open
Brillouin zone $(-\pi,\pi] \times (-\pi,\pi]$, there are only four
stationary points, which are commonly labeled as $\Gamma$, $X$,
$X'$, and $M$.

\begin{itemize}
\item[$\Gamma$:] $(k,p) = (0,0)$ is the minimum point of the dispersion surface with $\omega(0,0) = 0$;

\item[$X$:] $(k,p) = (0,\pi)$ is a saddle point of the dispersion surface with $\omega(0,\pi) = 4 \epsilon$;

\item[$X'$:] $(k,p) = (\pi,0)$ is the other saddle point of the dispersion surface with $\omega(\pi,0) = 4 \epsilon$;

\item[$M$:] $(k,p) = (\pi,\pi)$ is the maximum point of the dispersion surface with $\omega(\pi,\pi) = 8 \epsilon$.
\end{itemize}

Discrete line solitons may bifurcate from any stationary point
provided that the effective continuous NLS equation is focusing
\cite[Section 1.1.2]{Pel-book}. Let us consider each of the
possibilities. For definiteness, we assume that the line soliton is
localized along the $m$-direction and uniform along the
$n$-direction.

\begin{itemize}
\item[$\Gamma$:] For $(k,p) = (0,0)$, we substitute $u_{m,n}(t) = e^{i \mu^2 t} \psi_m$ and obtain
the stationary 1D dNLS equation
\begin{equation}
\label{1D-DNLS}
- \mu^2 \psi_m + \epsilon (\psi_{m+1} + \psi_{m-1} - 2 \psi_m) + |\psi_m|^2 \psi_m = 0,
\end{equation}
which admits discrete solitons for any $\epsilon > 0$ and $\mu
\ne 0$ \cite{Hermann,QinXiao}. Moreover, for fixed $\epsilon >
0$, the fundamental discrete soliton is approximated by the NLS
soliton
\begin{equation} \label{line-soliton-cont}
\psi_m \to \sqrt{2} \hspace{0.04cm} \mu \; {\rm sech}\left(\frac{\mu m}{\sqrt{\epsilon}}\right) \quad \mbox{\rm as} \quad \mu \to 0.
\end{equation}
This approximation was rigorously justified in the recent work
\cite{Bambusi} (see also \cite[Section 2.3.2]{Pel-book}).

\item[$X$:] For $(k,p) = (0,\pi)$, we substitute $u_{m,n}(t) = (-1)^n e^{i (\mu^2 - 4 \epsilon) t} \psi_m$ and obtain
the same  stationary dNLS equation (\ref{1D-DNLS}), which admits
the discrete solitons for any $\epsilon > 0$ and $\mu \ne 0$.

\item[$X'$:] For $(k,p) = (\pi,0)$, we substitute
$u_{m,n}(t) = (-1)^m e^{i (-\mu^2 - 4 \epsilon) t} \psi_m$ and obtain
the stationary 1D dNLS equation
\begin{equation} \label{1D-DNLS-bad}
\mu^2 \psi_m - \epsilon (\psi_{m+1} + \psi_{m-1} - 2 \psi_m) + |\psi_m|^2 \psi_m = 0.
\end{equation}
This stationary equation admits no discrete solitons for any
$\epsilon > 0$ \cite[Lemma 3.10]{Pel-book}. Indeed, by
projecting Eq. (\ref{1D-DNLS-bad}) to $\psi$ and denoting
$(\Delta \psi)_m \equiv \psi_{m+1} + \psi_{m-1} - 2 \psi_m$, we
obtain a contradiction
$$
\mu^2 \| \psi \|^2_{l^2} + \epsilon \langle \psi,(-\Delta) \psi \rangle + \| \psi \|^4_{l^4} = 0,
$$
where each term on the left side is positive definite.

\item[$M$:] For $(k,p) = (\pi,\pi)$, we substitute $u_{m,n}(t) = (-1)^{m+n} e^{i (-\mu^2 - 8 \epsilon) t} \psi_m$
and obtain the same stationary dNLS equation
(\ref{1D-DNLS-bad}), which admits no discrete solitons for any
$\epsilon > 0$.
\end{itemize}

From the above analysis, we see that only two bifurcations of
fundamental discrete line solitons occur and the bifurcation
points are $\Gamma$ and $X$. In the absence of transverse perturbations,
these fundamental line solitons are stable. In the following, we will analyze
transverse stability of these fundamental line solitons in the
anti-continuum limit $\epsilon \to 0$ for fixed $\mu
> 0$ (or equivalently, $\mu \to \infty$ for fixed $\epsilon > 0$).

Before the transverse-stability analysis in the anti-continuum
limit, it is useful to recall the transverse-stability results in
the opposite (continuum) limit that arises when $\epsilon \to
\infty$ for fixed $\mu > 0$ (or equivalently, $\mu \to 0$ for fixed
$\epsilon > 0$).

\begin{itemize}
\item[$\Gamma$:] For $(k,p) = (0,0)$, we substitute
$$
u_{m,n}(t) =\hspace{0.04cm} U(X,Y,t) e^{i \mu^2 t}, \  X = \frac{m}{\sqrt{\epsilon}}, \ Y = \frac{n}{\sqrt{\epsilon}}
$$
into Eq. (\ref{dNLS}). Assuming smoothness of the envelope
function $U(X,Y,t)$, we obtain an elliptic 2D NLS
equation for $U(X,Y,t)$ as $\epsilon \to \infty$:
\begin{equation}
\label{elliptic-NLS}
i \frac{\partial U}{\partial t} + \frac{\partial^2 U}{\partial X^2}
+ \frac{\partial^2 U}{\partial Y^2} + (|U|^2 - \mu^2) U = 0.
\end{equation}
The line soliton (\ref{line-soliton-cont}) is transversely
unstable in this elliptic NLS equation (\ref{elliptic-NLS}) due
to neck-type instability (see \cite{KivPel,Yang_SIAM} and
references therein).

\item[$X$:] For $(k,p) = (0,\pi)$, we substitute
$$
u_{m,n}(t) = (-1)^n U(X,Y,T) e^{i (\mu^2 - 4 \epsilon) t}, \ X = \frac{m}{\sqrt{\epsilon}}, \ Y = \frac{n}{\sqrt{\epsilon}}
$$
into Eq. (\ref{dNLS}). Assuming smoothness of the envelope
function $U(X,Y,t)$, we obtain a hyperbolic 2D NLS
equation for $U(X,Y,t)$ as $\epsilon \to \infty$:
\begin{equation}
\label{hyperbolic-NLS}
i \frac{\partial U}{\partial t} + \frac{\partial^2 U}{\partial X^2}
- \frac{\partial^2 U}{\partial Y^2} + (|U|^2 - \mu^2) U = 0.
\end{equation}
The line soliton (\ref{line-soliton-cont}) is also transversely
unstable in this hyperbolic NLS equation (\ref{hyperbolic-NLS})
due to snaking-type instability (see \cite{DecPel,Yang_SIAM} and
references therein).
\end{itemize}

From reductions to 2D NLS equations (\ref{elliptic-NLS}) and (\ref{hyperbolic-NLS}),
we see that discrete line solitons
(\ref{line-soliton-cont}) are always transversely unstable
in the continuum limit. Thus it
is surprising that discrete line solitons were reported to be
transversely stable far from the continuum limit when they bifurcate
from the X point of the dispersion surface \cite{Yang1}.
Line solitons bifurcated from the $\Gamma$ point, however, remain
transversely unstable for all values of $\epsilon$ (i.e.,
both near the continuum limit and away from it) \cite{Yang1}. Below
we shall prove these numerical observations by rigorous
spectral-stability analysis that relies on the count of eigenvalues
of negative energy \cite{ChPel,KKS,P05}. In addition, asymptotic
expressions for unstable eigenvalues will also be derived in the anti-continuum limit.

\subsection{Instability of line solitons bifurcating from the $\Gamma$ point}

Discrete line solitons bifurcating from the $\Gamma$ point are of
the form
\begin{equation} \label{s:gamma}
u_{m,n}(t) = e^{i \mu^2 t} \psi_m,
\end{equation}
where $\psi$ satisfies the stationary 1D dNLS equation
(\ref{1D-DNLS}). It can be easily shown that $\{ \psi_m \}_{m \in
\mathbb{Z}}$ in these discrete solitons is real-valued (up to
multiplication by $e^{i \alpha}$ for real $\alpha$, i.e., $\alpha
\in \mathbb{R}$) \cite[Lemma 3.11]{Pel-book}. Perturbing these line
solitons as
$$
u_{m,n}(t) = e^{i \mu^2 t} \left[ \psi_m + v_{m,n}(t) \right],
$$
and substituting it into the dNLS equation (\ref{dNLS}), we obtain
the linearized dNLS equation as
\begin{eqnarray*}
&& \hspace{-1.3cm} i \frac{d v_{m,n}}{d t} - \mu^2 v_{m,n} + \epsilon
(v_{m+1,n} + v_{m-1,n} + v_{m,n+1} + v_{m,n-1} - 4 v_{m,n})
+ \psi_m^2 (2 v_{m,n} + \bar{v}_{m,n}) = 0.
\end{eqnarray*}
For normal modes
\begin{equation} \label{normal-mode}
v_{m,n}(t) = e^{\lambda t + i p n} \left( U_{m} + i W_{m} \right),
\hspace{0.2cm} \bar{v}_{m,n}(t) = e^{\lambda t + i p n} \left( U_{m} - i W_{m} \right),
\end{equation}
we obtain the standard form of the eigenvalue problem
\begin{equation}
\label{lin-eigen}
L_+(p) U = -\lambda W, \quad L_-(p) W = \lambda U,
\end{equation}
where $L_{\pm}(p)$ are $p$-dependent 1D discrete Schr\"{o}dinger
operators,
\begin{eqnarray}  \label{e:LpmGamma}
\begin{array}{l}
\hspace{-0.7cm}
(L_+(p) U)_m \equiv - \epsilon
\left[U_{m+1} + U_{m-1} + 2 \cos(p) U_m - 4 U_{m}\right]  + \mu^2 U_m - 3 \psi_m^2 U_m,   \\
\hspace{-0.7cm}
(L_-(p) W)_m \equiv - \epsilon
\left[W_{m+1} + W_{m-1} + 2 \cos(p) W_m - 4 W_{m}\right]  + \mu^2 W_m - \psi_m^2 W_m. \end{array}
\end{eqnarray}
It is easy to see that eigenvalues $\lambda$ in the above
linear-stability problem always appear as quadruples $(\lambda,
\bar{\lambda}, -\lambda, -\bar{\lambda})$ when $\lambda$ is complex
or as pairs $(\lambda, -\lambda)$ when $\lambda$ is real or purely
imaginary.

Among the two parameters $\mu$ and $\epsilon$ in the above
eigenvalue problem, the ratio $\epsilon/\mu^2$ is invariant with
respect to a scaling transformation. The anti-continuum limit
corresponds to the limit of $\epsilon/\mu^2 \to 0$. Without loss of
generality, we fix $\mu = 1$ and consider small values of
$\epsilon>0$ below.

We are interested in transverse stability of the {\em fundamental}
line soliton $\psi_m$, which is positive for all $m \in \mathbb{Z}$
and confined to a single lattice site, say at $m = 0$, in the
anti-continuum limit $\epsilon \to 0$. Because the stationary
equation (\ref{1D-DNLS}) is analytic in $\epsilon$ and polynomial in
$\psi$, whereas the difference operator is bounded, the dependence
of $\psi$ on $\epsilon$ is real analytic near $\epsilon = 0$
\cite[Theorem 3.8]{Pel-book}. Using the regular perturbation method,
we can easily obtain the power series expansion for $\psi$ as
\begin{equation}
\label{expansion-1}
\psi_m = \delta_{m,0} + \epsilon (\delta_{m,1} + \delta_{m,0} + \delta_{m,-1} ) +
\mathcal{O}(\epsilon^2),
\end{equation}
where $\delta_{m,m'}$ is the Kronecker notation with $\delta_{m,m'}
= 1$ for $m = m'$ and 0 otherwise.

We shall now present the instability theorem for fundamental
discrete line solitons bifurcating from the $\Gamma$ point. This
fundamental line soliton exists for any $\epsilon > 0$
\cite{Hermann,QinXiao} (see also \cite[Theorem 3.12]{Pel-book}). Our
instability theorem below applies to all values of $\epsilon > 0$,
except that the asymptotic expression for the unstable eigenvalue is
valid only near the anti-continuum limit $\epsilon \to 0$.

\begin{theorem}
\label{theorem-1} Consider the fundamental discrete line soliton
(\ref{s:gamma}) bifurcating from the $\Gamma$ point in the dNLS
equation (\ref{dNLS}). For any $\epsilon > 0$, there is
$p_0(\epsilon) \in (0,\pi]$ such that for any $p \in
(-p_0(\epsilon), p_0(\epsilon)) \backslash\{0\}$ the
linear-stability problem (\ref{lin-eigen}) admits a symmetric pair
of real eigenvalues $\pm \lambda(\epsilon,p)$ with
$\lambda(\epsilon,p) > 0$. Hence this fundamental line soliton is
transversely unstable for all $\epsilon > 0$.  In addition,
$p_0(\epsilon)=\pi$ if $0<\epsilon<\frac{1}{2}$. Furthermore, for any
$p \in [-\pi,\pi]$, the eigenvalue $\lambda(\epsilon,p)$ has the
following asymptotic expansion in the anti-continuum limit,
\begin{equation}
\label{result-theorem-1}
\lambda^2(\epsilon,p) = 8 \epsilon \sin^2\left(\frac{p}{2}\right) + \mathcal{O}(\epsilon^2) \hspace{0.2cm}
\mbox{\rm as} \hspace{0.2cm} \epsilon \to 0.
\end{equation}
\end{theorem}

\begin{proof}
We first rewrite operators $L_{\pm}(p)$ in (\ref{e:LpmGamma}) as
$$
L_{\pm}(p) = L_{\pm}(0) + 2 \epsilon \left[1 - \cos(p)\right].
$$
These are bounded operators from $l^2(\mathbb{Z})$ to $l^2(\mathbb{Z})$,
which have both continuous and discrete spectra.

The stationary equation (\ref{1D-DNLS}) is simply $L_-(0) \psi = 0$.
Because $\psi$ is positive, $0$ is at the bottom of spectrum of
$L_-(0)$, so that $L_-(0)$ is non-negative \cite{Sturm1,Sturm2}.
By the perturbation theory, $L_-(p)$ is strictly positive for any $p \in [-\pi, \pi]
\backslash\{0\}$ and $\epsilon > 0$. On the other hand, $L_+(0)$ has
at least one negative eigenvalue because
$$
\langle L_+(0) \psi, \psi \rangle = - 2 \| \psi \|_{l^4}^4 < 0,
$$
where $\| \psi \|_{l^4}^4 = \sum_{n \in \mathbb{Z}} |\psi_n|^4$.
Moreover, in the limit $\epsilon \to 0$, only one negative
eigenvalue of $L_+(0)$ exists, which is the eigenvalue $-2$
associated with the central site $m = 0$. By the variational
arguments \cite{Hermann}, this negative eigenvalue persists and
remains the only negative eigenvalue of $L_+(0)$ for any $\epsilon >
0$. Since $L_+(p) \geq L_+(0)$, $L_+(p)$ has at most one negative
eigenvalue and no zero eigenvalues. It follows from the stationary
equation (\ref{1D-DNLS}) with $\mu = 1$ that
$$
\| \psi \|_{l^4}^4 = \| \psi \|^2_{l^2} + \epsilon \langle \psi,(-\Delta) \psi \rangle \geq \| \psi \|^2_{l^2},
$$
where $\Delta$ is the 1D discrete Laplacian. Thus we obtain
\begin{eqnarray*}
\langle L_+(p) \psi, \psi \rangle  =    - 2 \| \psi \|_{l^4}^4 + 2 \epsilon \left[ 1 - \cos(p)\right] \| \psi \|^2_{l^2}
  \leq 2 \left\{\epsilon \left[ 1 - \cos(p)\right] - 1\right\} \| \psi \|^2_{l^2},
\end{eqnarray*}
hence $L_+(p)$ admits a negative eigenvalue for any $p \in
[-\pi,\pi]$ if $0<\epsilon < \frac{1}{2}$ and for at least small $p$
if $\epsilon > 0$ is arbitrary. In other words, for any $\epsilon >
0$, there is $p_0(\epsilon) \in (0,\pi]$ such that $L_+(p)$ has a
negative eigenvalue for any $p \in (-p_0(\epsilon), p_0(\epsilon))$.
Moreover, $p_0(\epsilon) = \pi$ at least for $0<\epsilon < \frac{1}{2}$.

For $p = 0$, the linear eigenvalue problem (\ref{lin-eigen}) admits
zero eigenvalue of algebraic multiplicity two for any $\epsilon > 0$
because $L_-(0)\psi=0$ and
$$
L_+(0) \frac{\partial \psi}{\partial (\mu^2)} = -\psi.
$$
This zero eigenvalue is destroyed when $p \neq 0$ and this may cause
instability when splitting of this double zero eigenvalue occurs
along the real axis. Using the negative index theory \cite{ChPel,KKS,P05,Pel-book},
we obtain:
\begin{equation}
\label{count-1}
\begin{array}{l}
N_{\rm real}^- + N_{\rm imag}^- + N_{\rm comp} = n(L_+(p)), \\
N_{\rm real}^+ + N_{\rm imag}^- + N_{\rm comp} = n(L_-(p)), \end{array}  \hspace{0.15cm}
p \in [-\pi, \pi] \backslash\{0\},
\end{equation}
where $N_{\rm real}^+$ ($N_{\rm real}^-$) are the numbers of real
positive eigenvalues $\lambda$ with positive (negative) quadratic
form $\langle L_+(p) U,U \rangle$ at the eigenvector $(U,W)$ of the
eigenvalue problem (\ref{lin-eigen}), $N_{\rm imag}^-$ is the number
of purely imaginary eigenvalues $\lambda$ with ${\rm Im}(\lambda) >
0$ and negative quadratic form $\langle L_+(p) U,U \rangle$, and
$N_{\rm comp}$ is the number of complex eigenvalues $\lambda$ with
${\rm Re}(\lambda) > 0$ and ${\rm Im}(\lambda) > 0$, counting their
algebraic multiplicities. Note that eigenvalues contributing to $N_{\rm imag}^-$ are called
eigenvalues with a negative Krein signature \cite{KKS}. The
eigenvalue-counting formula (\ref{count-1}) follows
directly from Theorem 4.5 of \cite{Pel-book} because operators $L_{\pm}(p)$
have no zero eigenvalues for any $p \neq 0$.

The preceding computations show that $n(L_-(p)) = 0$, and $n(L_+(p))
= 1$ for $p \in (-p_0(\epsilon), p_0(\epsilon))$. In these cases,
the index formula (\ref{count-1}) yields
$$
N_{\rm real}^- = 1, \quad N_{\rm real}^+ = N_{\rm imag}^- = N_{\rm comp} = 0,
$$
which proves the statement of Theorem 1 on transverse instability.
It remains to justify the asymptotic expansion
(\ref{result-theorem-1}) for the real positive eigenvalue
$\lambda(\epsilon,p)$ as $\epsilon \to 0$.

When $\epsilon = 0$, the spectral problem (\ref{lin-eigen}) with
$\mu = 1$ has three points in the spectrum: $\lambda = 0$ of
algebraic multiplicity two and $\lambda = \pm i$ of infinite
algebraic multiplicity. Continuous spectral bands bifurcate from the
points $\lambda = \pm i$ for $\epsilon \neq 0$. This bifurcation was
studied in detail in the recent work \cite{PelSak}, and no unstable
eigenvalues arise in this bifurcation. We shall now calculate the
splitting of the double zero eigenvalue for any fixed $p
> 0$ and small $\epsilon > 0$, using the expansion
(\ref{expansion-1}) near the anti-continuum limit.

We rewrite the eigenvalue problem
(\ref{lin-eigen}) with $\mu = 1$ at the central site $m = 0$ as follows:
\begin{eqnarray*}
\begin{array}{l}
\hspace{-0.5cm}
2 U_0 + \epsilon \left[U_{1} + U_{-1} + 2 \cos(p) U_0 + 2 U_0 \right] + \mathcal{O}(\epsilon^2) U_0 = \lambda W_0, \\
\hspace{-0.5cm}
\epsilon \left[W_{1} + W_{-1} + 2 \cos(p) W_0 - 2 W_0\right] + \mathcal{O}(\epsilon^2) W_0 = -\lambda U_0.\end{array}
\end{eqnarray*}
By using the scaling transformation $U = \sqrt{\epsilon}
\hspace{0.06cm} \mathcal{U}$, $W = \mathcal{W}$, and $\lambda =
\sqrt{\epsilon} \hspace{0.04cm}  \Lambda$, the system can be
rewritten in the equivalent form:
\begin{eqnarray*}
\hspace{-0.75cm} \begin{array}{l}
2 \mathcal{U}_0 + \epsilon (\mathcal{U}_{1} + \mathcal{U}_{-1} + 2 \cos(p) \mathcal{U}_0 + 2 \mathcal{U}_0 ) +
\mathcal{O}(\epsilon^2) \mathcal{U}_0 = \Lambda \mathcal{W}_0, \\
\mathcal{W}_{1} + \mathcal{W}_{-1} + 2 \cos(p) \mathcal{W}_0 - 2 \mathcal{W}_0
+ \mathcal{O}(\epsilon) \mathcal{W}_0 = - \Lambda \mathcal{U}_0.\end{array}
\end{eqnarray*}
At the adjacent sites $m = \pm 1$, the linear eigenvalue problem
(\ref{lin-eigen}) is
\begin{eqnarray*}
\hspace{-0.7cm} \begin{array}{l}
\mathcal{U}_{\pm 1} - \epsilon \left[\mathcal{U}_{\pm 2} + \mathcal{U}_{0} + 2 \cos(p) \mathcal{U}_{\pm 1} - 4 \mathcal{U}_{\pm 1} \right]   + \mathcal{O}(\epsilon^2) \mathcal{U}_{\pm 1} = -\Lambda \mathcal{W}_{\pm 1}, \\
\mathcal{W}_{\pm 1} - \epsilon \left[\mathcal{W}_{\pm 2} + \mathcal{W}_{0} + 2 \cos(p) \mathcal{W}_{\pm 1} - 4 \mathcal{W}_{\pm 1} \right]  + \mathcal{O}(\epsilon^2) \mathcal{W}_{\pm 1} = \epsilon \Lambda \mathcal{U}_{\pm 1},\end{array}
\end{eqnarray*}
since $\psi^2_{\pm 1} = \mathcal{O}(\epsilon^2)$. Similar equations
can be written for any $m \neq 0$.

For $\Lambda = \mathcal{O}(1)$, we have the reduction
$\mathcal{U}_{\pm m} = \mathcal{O}(\epsilon^m) \mathcal{U}_0$ and
$\mathcal{W}_{\pm m} = \mathcal{O}(\epsilon^m) \mathcal{W}_0$ for
any $m \in \mathbb{N}$, which enables us to close the leading-order
equations for $(\mathcal{U}_0, \mathcal{W}_0)$:
\begin{eqnarray*}
\begin{array}{l} 2 \mathcal{U}_0 + \mathcal{O}(\epsilon) \mathcal{U}_0 = \Lambda \mathcal{W}_0, \\
\left[2 \cos(p) - 2 \right] \mathcal{W}_0 + \mathcal{O}(\epsilon) \mathcal{W}_0 = - \Lambda \mathcal{U}_0.\end{array}
\end{eqnarray*}
After eliminating $\mathcal{W}_0$, we obtain the algebraic equation
for $\Lambda$ as
\begin{equation*}
\Lambda^2 = 2 (2 - 2 \cos(p)) + \mathcal{O}(\epsilon) = 8 \sin^2\left(\frac{p}{2}\right) + \mathcal{O}(\epsilon),
\end{equation*}
which then yields the asymptotic expansion (\ref{result-theorem-1}).
\end{proof}

\begin{remark}
For any fixed $p \in [-\pi, \pi] \backslash \{0\}$, we obtain
$\Lambda^2 > 0$ as $\epsilon \to 0$, which guarantees spectral
instability of these discrete line solitons for small $\epsilon >
0$. Note that the asymptotic formula (\ref{result-theorem-1}) is not
uniform as $\epsilon \to 0$ and $p \to 0$, and a different
perturbation theory is needed in the limit $p\to 0$ for fixed
$\epsilon > 0$ (see \cite{KivPel,Yang_SIAM} and references therein).
\end{remark}

\begin{remark}
Using the resolvent analysis from \cite{PelSak}, one can show that
the continuous spectral bands are the two line-segments on the
imaginary axis,
\begin{eqnarray*}
 i \lambda \in [-1-2 \epsilon (3 -
\cos(p)),-1-2\epsilon (1 - \cos(p))] \,
\cup  \,  [1+2\epsilon (1 - \cos(p)),1+2 \epsilon (3 - \cos(p))],
\end{eqnarray*}
whereas no discrete (isolated) eigenvalues bifurcate out from
the point $\lambda = \pm i$ as $\epsilon \neq 0$.
\end{remark}

\begin{remark}
In the continuous limit $\epsilon \to +\infty$, the discrete line
solitons are asymptotically described by the elliptic NLS equation
(\ref{elliptic-NLS}), where unstable real eigenvalues are restricted
to the $p$-interval $(-p_0(\epsilon),p_0(\epsilon) \backslash \{0\}$
with $p_0(\epsilon \to \infty) = \sqrt{3}$ (for $\mu=1$) \cite{KivPel,Yang_SIAM}.
Therefore, $p_0(\epsilon) < \pi$ for sufficiently large positive $\epsilon$.
\end{remark}

\subsection{Stability of line solitons bifurcating from the $X$ point}

Discrete line solitons bifurcating from the $X$ point are of the
form
\begin{equation} \label{s:xpoint}
u_{m,n}(t) = (-1)^n e^{i (\mu^2 - 4 \epsilon) t} \psi_m,
\end{equation}
where $\psi$ is a real-valued solution of the stationary 1D dNLS
equation (\ref{1D-DNLS}). Notice that these solitons at adjacent
lattice sites along the transverse $n$-direction are out-of-phase
with each other, which contrasts the line solitons bifurcating
from the $\Gamma$ point, where the solitons at adjacent lattice sites along
the transverse direction are in-phase with each other.

Linearizing the dNLS equation (\ref{dNLS}) around this solution, we
substitute
$$
u_{m,n}(t) = (-1)^n e^{i (\mu^2 - 4 \epsilon) t} \left[ \psi_m + v_{m,n}(t) \right],
$$
and obtain the linearized dNLS equation
\begin{eqnarray*}
\hspace{-0.4cm}
i \frac{d v_{m,n}}{d t} - \mu^2 v_{m,n} + \epsilon
(v_{m+1,n} + v_{m-1,n} - v_{m,n+1} - v_{m,n-1})   +
\psi_m^2 (2 v_{m,n} + \bar{v}_{m,n}) = 0.
\end{eqnarray*}
For normal modes (\ref{normal-mode}), we obtain the eigenvalue
problem
\begin{equation}
\label{lin-eigen2}
L_+(p) U = -\lambda W, \quad L_-(p) W = \lambda U,
\end{equation}
where
\begin{eqnarray}  \label{e:LpmX}
\begin{array}{l}
(L_+(p) U)_m \equiv - \epsilon
\left[U_{m+1} + U_{m-1} - 2 \cos(p) U_m\right]  + \mu^2 U_m - 3 \psi_m^2 U_m, \\
(L_-(p) W)_m \equiv - \epsilon
\left[W_{m+1} + W_{m-1} - 2 \cos(p) W_m\right]  + \mu^2 W_m - \psi_m^2 W_m. \end{array}
\end{eqnarray}

Again, we are interested in transverse stability of the fundamental
line soliton $\psi$ which is positive and given by the power series
expansion (\ref{expansion-1}) in the anti-continuum limit
$\epsilon\to 0$. As before, we will fix $\mu=1$ without loss of
generality. The next theorem guarantees stability of this
fundamental discrete line soliton for small values of $\epsilon$.

\begin{theorem}
\label{theorem-2} Consider the fundamental discrete line soliton
(\ref{s:xpoint}) bifurcating from the $X$ point in the dNLS equation
(\ref{dNLS}). There exists $\epsilon_0 > 0$ such that for any
$\epsilon \in (0,\epsilon_0)$ and $p \in [\pi,\pi]$, the
linear-stability problem (\ref{lin-eigen}) does not admit
any unstable eigenvalues, thus the fundamental line soliton for
small values of $\epsilon$ is transversely stable. This stable line
soliton possesses a pair of discrete imaginary eigenvalues $\pm i
\omega(\epsilon,p)$ of negative Krein signature. Moreover, for any
$p \in [-\pi,\pi]$ and small $\epsilon$, this eigenvalue
$\omega(\epsilon,p)$ has the following asymptotic expression,
\begin{equation}
\label{result-theorem-2}
\omega^2(\epsilon,p) = 8 \epsilon \sin^2\left(\frac{p}{2}\right) + \mathcal{O}(\epsilon^2) \hspace{0.2cm}
\mbox{\rm as} \hspace{0.2cm} \epsilon \to 0.
\end{equation}
\end{theorem}

\begin{proof}
We first rewrite operators $L_{\pm}(p)$ in (\ref{e:LpmX}) as
$$
L_{\pm}(p) = L_{\pm}(0) - 2 \epsilon \left[1 - \cos(p)\right].
$$
Because $L_-(0) \psi = 0$ and $\psi$ is positive, $0$ is the lowest
eigenvalue of $L_-(0)$ for any $\epsilon > 0$ \cite{Sturm1,Sturm2}.
By perturbation theory, $L_-(p)$ has exactly one negative eigenvalue for any $p
\in [-\pi, \pi] \backslash\{0\}$ and small positive $\epsilon$. On
the other hand, since $\psi$ and $L_+(0)$ for the line
soliton (\ref{s:xpoint}) are the same as those for the
line soliton (\ref{s:gamma}), the variational arguments from \cite{Hermann}
imply that $L_+(0)$ has exactly one
negative eigenvalue and no zero eigenvalue for any $\epsilon > 0$.
Therefore, $L_+(p)$ has a single negative eigenvalue for any $p \in
[-\pi, \pi]$ and small positive $\epsilon$.

For any $p \neq 0$, we again use the eigenvalue-counting formula
(\ref{count-1}), which equally applies to the linear eigenvalue
problem (\ref{lin-eigen2}). The preceding computation shows that
there is $\epsilon_0 > 0$ such that $n(L_-(p)) = 1$ and $n(L_+(p)) =
1$ for any $\epsilon \in (0,\epsilon_0)$ and $p \in [-\pi,\pi]
\backslash \{0\}$. Since eigenvalues in the spectral problem
(\ref{lin-eigen2}) appear as quadruples $(\lambda, \bar{\lambda},
-\lambda, -\bar{\lambda})$ for complex $\lambda^2$ and as pairs $\pm
\lambda$ for real $\lambda^2$ and since the zero eigenvalue for $p =
0$ has algebraic multiplicity two, this zero eigenvalue splits along
the real or imaginary axis as a pair of simple eigenvalues for $p
\neq 0$. Combining this with the eigenvalue-counting formula
(\ref{count-1}), we easily see that this splitting occurs along the
imaginary axis, and for any $\epsilon \in (0,\epsilon_0)$ and $p \in
[-\pi,\pi] \backslash \{0\}$,
$$
N_{\rm imag}^- = 1, \quad N_{\rm real}^+ = N_{\rm real}^- = N_{\rm comp} = 0,
$$
which proves the transverse-stability statement in Theorem 2. Note
that the imaginary eigenvalues of negative Krein signature persist
on the imaginary axis, unless they coalesce with other eigenvalues
of positive Krein signature or continuous spectral bands.

Next we prove the asymptotic expansion (\ref{result-theorem-2}) for
the imaginary eigenvalue $i \omega(\epsilon,p)$ as $\epsilon \to 0$.
When $\epsilon = 0$, the spectral problem (\ref{lin-eigen}) with
$\mu = 1$ has three points in the spectrum: $\lambda = 0$ of
algebraic multiplicity two and $\lambda = \pm i$ of infinite
algebraic multiplicity. For small $\epsilon$, we only need to
compute the splitting of the double zero eigenvalue for any fixed $p
\in [-\pi,\pi]$, using the expansion (\ref{expansion-1}) near the
anti-continuum limit.

Repeating the perturbation expansions and using the scaling
transformation $U = \sqrt{\epsilon} \hspace{0.06cm} \mathcal{U}$, $W
= \mathcal{W}$, and $\lambda = \sqrt{\epsilon}\hspace{0.06cm}
\Lambda$, we obtain the linear eigenvalue problem at the central
site $m = 0$:
\begin{eqnarray*}
\hspace{-0.7cm}
2 \mathcal{U}_0 + \epsilon \left[\mathcal{U}_{1} + \mathcal{U}_{-1} - 2 \cos(p) \mathcal{U}_0 + 6 \mathcal{U}_0 \right] +
\mathcal{O}(\epsilon^2) \mathcal{U}_0  =  \Lambda \mathcal{W}_0, \\
\hspace{-0.7cm}
\mathcal{W}_{1} + \mathcal{W}_{-1} - 2 \cos(p) \mathcal{W}_0 + 2 \mathcal{W}_0
+ \mathcal{O}(\epsilon) \mathcal{W}_0  =  - \Lambda \mathcal{U}_0.
\end{eqnarray*}
Similar to the proof of Theorem \ref{theorem-1}, for $\Lambda =
\mathcal{O}(1)$, we have the reduction $\mathcal{U}_{\pm m} =
\mathcal{O}(\epsilon^m) \mathcal{U}_0$ and $\mathcal{W}_{\pm m} =
\mathcal{O}(\epsilon^m) \mathcal{W}_0$ for any $m \in \mathbb{N}$,
hence the above equations yield
\begin{equation*}
\Lambda^2 = -2 \left[2 - 2 \cos(p)\right] + \mathcal{O}(\epsilon) = -8 \sin^2\left(\frac{p}{2}\right) + \mathcal{O}(\epsilon),
\end{equation*}
which yields the asymptotic expansion (\ref{result-theorem-2}).
\end{proof}

\begin{remark}
For any $p \in [-\pi,\pi] \backslash \{0\}$ and small values of
$\epsilon$, we get $\Lambda^2 < 0$, which gives imaginary
eigenvalues $\pm i \omega(\epsilon,p)$. It is easy to see that these
imaginary eigenvalues have negative Krein signature, meaning that
the quadratic form $\langle L_+(p) U, U \rangle$ at the eigenvector
$(U,W)$ is negative. For small values of $\epsilon$, these imaginary
eigenvalues are bounded away from the continuous spectrum
bifurcating out of the points $\pm i$, which guarantees spectral
stability of discrete line solitons for small positive $\epsilon$.
\end{remark}

\begin{remark}
Using the resolvent analysis from \cite{PelSak}, one can show that
the continuous spectral bands are located at the two segments on the
imaginary axis:
\begin{eqnarray*} \label{bandX}
i \lambda \in [-1-2 \epsilon (1+\cos(p)),-1+2\epsilon (1 - \cos(p))] \,  \cup  \,
[1-2\epsilon (1 - \cos(p)),1+2 \epsilon (1+\cos(p))],
\end{eqnarray*}
and no discrete (isolated) eigenvalues bifurcate out from the points
$\lambda = \pm i$ as $\epsilon \to 0$.
\end{remark}

\begin{remark}
In the continuum limit $\epsilon \to +\infty$, discrete line
solitons (\ref{s:xpoint}) from the X point in Eq. (\ref{dNLS}) are
asymptotically described by the hyperbolic NLS equation
(\ref{hyperbolic-NLS}), where line solitons are transversely
unstable for any nonzero transverse wave number $p$
\cite{DecPel,Yang_SIAM}. Hence unstable eigenvalues must appear for
these discrete line solitons at sufficiently large positive
$\epsilon$. These unstable eigenvalues can appear through collisions
of imaginary eigenvalues of negative Krein signature with the
continuous spectral band or with additional imaginary eigenvalues of
positive Krein signature.
\end{remark}

\section{One-dimensional (stripe) lattice} \label{sec:1D}

In this section, we consider transverse stability of line solitons
in the dNLS equation on a one-dimensional lattice with continuous
transverse dispersion. The mathematical model for this problem is
\begin{equation}
\label{dNLS-diff}
i \frac{\partial u_{m}}{\partial t} + \epsilon (u_{m+1} + u_{m-1} - 2 u_{m})
+ \kappa \frac{\partial^2 u_m}{\partial y^2} + |u_{m}|^2 u_{m} = 0,
\end{equation}
where $m \in \mathbb{Z}$, the complex variable $u_{m}$ depends on
the evolution time $t$ and the transverse coordinate $y$. Here the
sign of nonlinearity has been normalized to be unity through a
scaling of $\epsilon$, $\kappa$ and $t$. By the staggering
transformation
\begin{equation}
u_{m}(y,t) = (-1)^{m} v_{m}(y,t) e^{-4 i \epsilon t},
\end{equation}
we can map the dNLS equation (\ref{dNLS-diff}) for $u$ with
$\epsilon < 0$ to the same equation for $v$ with $\epsilon > 0$.
Thus we set $\epsilon>0$ below but consider both positive and
negative values of the transverse dispersion parameter $\kappa$.
Through a scaling of $y$, we normalize $\kappa$ so that $\kappa=\pm
1$.

Transverse instability of line solitons was reported in \cite{Yang2}
for $\kappa = -1$ and any $\epsilon > 0$. We shall prove this
numerical observation by rigorous study of spectral stability. We
shall also study the case $\kappa = +1$ for completeness.

First, by inserting the discrete Fourier modes $u_{m}(y,t) = e^{i k
m  - i \omega t}$ into the linear dNLS equation (\ref{dNLS-diff}),
we find that the discrete-dispersion relation is
$$
\omega(k) = 2 \epsilon \left[1 -\cos(k)\right],
$$
where the wavenumber $k$ is in the first Brillouin zone $k\in [-\pi,
\pi]$. For $\epsilon > 0$, discrete line solitons bifurcate from the
minimum of this dispersion curve towards negative values of
$\omega$. Therefore the discrete line solitons are of the form
\begin{equation} \label{s:1D}
u_{m}(y,t) = e^{i \mu^2 t} \psi_m,
\end{equation}
where $\psi$ is a real-valued solution of the stationary 1D dNLS
equation (\ref{1D-DNLS}). Linearizing around this solution, we substitute
$$
u_{m}(y,t) = e^{i \mu^2 t} \left[ \psi_m + v_{m}(y,t) \right]
$$
into the dNLS equation (\ref{dNLS-diff}) and obtain the linearized dNLS equation
\begin{equation*}
i \frac{\partial v_{m}}{\partial t} - \mu^2 v_{m} + \epsilon
(v_{m+1} + v_{m-1} - 2 v_{m}) + \kappa \frac{\partial^2 v_m}{\partial y^2} +
\psi_m^2 (2 v_{m} + \bar{v}_{m}) = 0.
\end{equation*}
For the normal mode
\begin{equation*}
v_{m}(y,t) = e^{\lambda t + i p y} \left( U_{m} + i W_{m} \right),
\hspace{0.2cm} \bar{v}_{m}(y,t) = e^{\lambda t + i p y} \left( U_{m} - i W_{m} \right),
\end{equation*}
we obtain the linear-stability eigenvalue problem
\begin{equation} \label{lin-eigen3}
L_+(p) U = -\lambda W, \quad L_-(p) W = \lambda U,
\end{equation}
where
\begin{eqnarray} \label{e:Lpm3}
\begin{array}{l} (L_+(p) U)_m \equiv
- \epsilon
(U_{m+1} + U_{m-1} - 2 U_{m})  + (\mu^2 + \kappa p^2) U_m - 3 \psi_m^2 U_m, \\
(L_-(p) W)_m \equiv - \epsilon
(W_{m+1} + W_{m-1} - 2 W_{m})  + (\mu^2 + \kappa p^2) W_m - \psi_m^2 W_m. \end{array}
\end{eqnarray}
As before, we set $\mu = 1$ by variable rescaling and consider the fundamental
line soliton represented by the power series expansion (\ref{expansion-1}) for small $\epsilon$.

When $\epsilon=0$, the eigenvalue problem (\ref{lin-eigen3}) with
$\mu = 1$ has four points in the spectrum: two simple eigenvalues at
$\lambda = \pm \sqrt{\kappa p^2 (2 - \kappa p^2)}$ and two other
eigenvalues of infinite algebraic multiplicities at $\lambda=\pm
i(1+\kappa p^2)$.

If $\kappa =1$, the simple eigenvalues $\lambda=\pm p \sqrt{2 -
p^2}$ are real for $0<p^2<2$, thus the discrete line soliton
(\ref{s:1D}) is transversely unstable even in the unperturbed
($\epsilon=0$) case. These real eigenvalues persist for small
$\epsilon$. The following theorem shows that the transverse
instability of discrete line solitons (\ref{s:1D}) with $\kappa = 1$
holds for any $\epsilon> 0$.

\begin{theorem}
\label{theorem-3} Consider the fundamental discrete line soliton
(\ref{s:1D}) in the dNLS equation (\ref{dNLS-diff}) with $\kappa=1$.
For any $\epsilon > 0$, there is $p_0(\epsilon) >0$ such that for
any $p \in (-p_0(\epsilon),p_0(\epsilon)) \backslash\{0\}$ the
linear-stability problem (\ref{lin-eigen3}) admits a pair
of real eigenvalues $\pm \lambda(\epsilon,p)$ with
$\lambda(\epsilon,p)> 0$, thus this line soliton is transversely
unstable. In addition, for small $\epsilon$, $p_0(\epsilon)$ and
$\lambda(\epsilon,p)$ are given asymptotically by
\begin{equation}
\label{result-theorem-3}
p_0(\epsilon) = \sqrt{2} + \mathcal{O}(\epsilon), \hspace{0.15cm}
\lambda(\epsilon,p) = p \sqrt{2 - p^2} + \mathcal{O}(\epsilon) \hspace{0.15cm} \mbox{\rm as} \hspace{0.15cm}
\epsilon \to 0.
\end{equation}
\end{theorem}

\begin{proof}
We first rewrite operators $L_{\pm}(p)$ in (\ref{e:Lpm3}) as
$$
L_{\pm}(p) = L_{\pm}(0) + p^2.
$$
Similar to the proof of Theorem \ref{theorem-1}, we can see that
$L_-(p)$ is strictly positive for any $p \ne 0$ and $\epsilon \geq
0$. On the other hand, for any $\epsilon \geq 0$, $L_+(0)$ has a
single negative eigenvalue $-\beta$, where $\beta>0$. In particular,
when $\epsilon=0$, the negative eigenvalue with $\beta=2$ is
associated with the central site $m = 0$. Thus by denoting
$p_0(\epsilon)\equiv \sqrt{\beta}$, we see that $L_+(p)$ has exactly
one negative eigenvalue for $p \in (-p_0(\epsilon),p_0(\epsilon))$
and is strictly positive for any $|p|> p_0(\epsilon)$. The
eigenvalue-counting formula (\ref{count-1}) then yields that when $p
\in (-p_0(\epsilon),p_0(\epsilon))\backslash\{0\}$,
$$
N_{\rm real}^- = 1, \quad N_{\rm real}^+ = N_{\rm imag}^- = N_{\rm comp} = 0.
$$
The asymptotic expansions (\ref{result-theorem-3}) directly follow
from the preceding computations at $\epsilon = 0$ and the
analyticity of the linear eigenvalue problem (\ref{lin-eigen3}) in
$\epsilon$.
\end{proof}

\begin{remark}
The asymptotic expansion $\lambda(\epsilon,p) = p \sqrt{2 - p^2} +
\mathcal{O}(\epsilon)$ works equally well for $|p| > p_0(\epsilon)$,
where this $\lambda(\epsilon,p)$ is purely imaginary. These
imaginary eigenvalues have positive Krein signature and are bounded
away from the continuous spectrum located at
$$
i \lambda \in [-(1+p^2+4 \epsilon), -(1+p^2)] \cup [1+p^2,1+p^2+4\epsilon].
$$
\end{remark}

If $\kappa =-1$ and $\epsilon=0$, the simple eigenvalues $\lambda =
\pm i p \sqrt{2 + p^2}$ are purely imaginary, so are the eigenvalues
$\lambda=\pm i(1 - p^2)$ of infinite algebraic multiplicities. These
eigenvalue branches intersect at $p = \pm p_c$, where $p_c =
\frac{1}{2}$. When $p \neq \pm p_c$ and $0<\epsilon\ll 1$, the
simple eigenvalues persist on $i \mathbb{R}$, whereas two continuous
spectral bands bifurcate from the non-simple eigenvalues
$\lambda=\pm i(1 - p^2)$ along the two segments on $i \mathbb{R}$ as
$$
i \lambda \in [1-p^2, 1-p^2+4\epsilon] \cup [-(1-p^2+4\epsilon), -(1-p^2)].
$$
However, when $p = \pm p_c$ and $0<\epsilon\ll 1$, a resonance
occurs between these simple and non-simple eigenvalues, and as a
consequence, complex (unstable) eigenvalues bifurcate out. Notice
that the simple eigenvalues $\lambda = \pm i p \sqrt{2 + p^2}$ have
negative Krein signature, whereas the non-simple eigenvalues
$\lambda=\pm i(1 - p^2)$ have positive Krein signature. This
bifurcation of complex eigenvalues due to collision of eigenvalues
with opposite Krein signatures is a common phenomenon in Hamiltonian
systems \cite{KKS,VP}.

The following theorem guarantees instability of discrete line
solitons (\ref{s:1D}) in the dNLS equation (\ref{dNLS-diff}) with
$\kappa=-1$ for small values of $\epsilon
> 0$. This instability is caused by complex eigenvalues with small
real parts, and it occurs for intermediate values of transverse
wavenumbers $p$.

\begin{theorem}
\label{theorem-4} Consider the fundamental discrete line soliton
(\ref{s:1D}) in the dNLS equation (\ref{dNLS-diff}) with
$\kappa=-1$. There exists $\epsilon_0 > 0$ such that for any
$\epsilon \in (0,\epsilon_0)$, there exist $p_c^{\pm}(\epsilon)$
with ordering $0 < p_c^-(\epsilon) < p_c^+(\epsilon) < +\infty$, so
that for any $|p| \in (p_c^-(\epsilon),p_c^+(\epsilon))$ the
linear-stability problem (\ref{lin-eigen3}) admits a
quartet of complex eigenvalues $\pm \lambda(\epsilon,p)$, $\pm
\bar{\lambda}(\epsilon,p)$ with ${\rm Re} \lambda(\epsilon,p) > 0$
and ${\rm Im} \lambda(\epsilon,p)
> 0$. In addition, when $\epsilon \to 0$,
$p_c^{\pm}(\epsilon)$ and $\lambda(\epsilon,p)$ are given
asymptotically by
\begin{equation}
\label{result-theorem-5}
p_c^{\pm}(\epsilon) = \frac{1}{2} + \frac{\epsilon}{2} \left(1 \pm \frac{\sqrt{15}}{2}\right) +
\mathcal{O}(\epsilon^2),
\end{equation}
and
\begin{equation}
\label{result-theorem-4}
\lambda(\epsilon,p) = \frac{3}{4}i + \frac{i \epsilon}{15}(14 + 17 \delta)
+ \frac{2 \epsilon}{15} \sqrt{15 - 4(1-2\delta)^2} + \mathcal{O}(\epsilon^2),
\end{equation}
where $\delta \equiv \epsilon^{-1}(p^2 - \frac{1}{4}) =
\mathcal{O}(1)$. Furthermore, the most unstable eigenvalue
$\lambda_{max}(\epsilon)$ occurs at the transverse wavenumbers $\pm
p_{max}(\epsilon)$, where $\lambda_{max}(\epsilon)$ and
$p_{max}(\epsilon)$ are given by
\begin{equation}  \label{e:lammax}
\lambda_{max}(\epsilon)=\frac{3}{4} i + \epsilon \left(\frac{2}{\sqrt{15}} +\frac{3}{2}i\right)
+ \mathcal{O}(\epsilon^2), \hspace{0.2cm}
p_{max}=\frac{1}{2} + \frac{1}{2} \epsilon+
\mathcal{O}(\epsilon^2).
\end{equation}
\end{theorem}

\begin{proof}
Modifying the arguments from the proof of Theorem \ref{theorem-3}, we have now
$$
L_{\pm}(p) = L_{\pm}(0) - p^2.
$$
Therefore, for sufficiently small $\epsilon$, there is $p_0(\epsilon) > 0$ such that
the operators $L_{\pm}(p)$ have exactly one negative eigenvalue for all
$p \in (-p_0(\epsilon),p_0(\epsilon)) \backslash\{0\}$. Note that
$p_0(\epsilon) = 1 + \mathcal{O}(\epsilon)$ as $\epsilon \to 0$.

The eigenvalue-counting formula (\ref{count-1}) yields now
$$
N_{\rm imag}^- + N_{\rm comp} = 1, \hspace{0.2cm} N_{\rm real}^+ = N_{\rm real}^- = 0,
\hspace{0.2cm}   p \in (-p_0(\epsilon),p_0(\epsilon)\backslash\{0\}.
$$
The preceding computations and the analyticity of the linear
eigenvalue problem (\ref{lin-eigen3}) in $\epsilon$ imply that for
sufficiently small $\epsilon$, there are $p_c^{\pm}(\epsilon)$ with
ordering $0 < p_c^-(\epsilon) < p_c^+(\epsilon) < p_0(\epsilon)$
such that
$$
N_{\rm imag}^- = 1, \hspace{0.15cm} N_{\rm comp} = 0, \hspace{0.2cm} \mbox{for} \hspace{0.2cm} |p| \in (0,p_c^-(\epsilon)) \ \mbox{and} \ (p_c^+(\epsilon),p_0(\epsilon)),
$$
where $p_c^{\pm}(\epsilon) = \frac{1}{2} + \mathcal{O}(\epsilon)$ as
$\epsilon \to 0$. For these values of $\epsilon$ and $p$, the
discrete line solitons are spectrally stable. It remains to show
that
$$
N_{\rm imag}^- = 0 \hspace{0.2cm} \mbox{\rm and} \hspace{0.2cm}   N_{\rm comp} = 1 \quad \mbox{\rm for}
\quad |p| \in (p_c^-(\epsilon),p_c^+(\epsilon))
$$
due to a resonance between eigenvalues of negative and positive
Krein signatures.

First we introduce a scaling transformation
$$
p^2 = \frac{1}{4} + \epsilon \delta, \hspace{0.2cm} \lambda = \frac{3 i}{4} + i \epsilon \gamma, \hspace{0.2cm}
U_m = \frac{a_m + b_m}{2}, \hspace{0.2cm}  W_m = \frac{a_m - b_m}{2i},
$$
where $\delta, \gamma = \mathcal{O}(1)$. Under this transformation,
the eigenvalue problem (\ref{lin-eigen3}) for $\mu = 1$ and $\kappa
= -1$ becomes
\begin{eqnarray}
-\epsilon ( a_{m+1} - 2a_m + a_{m-1}) - (1+2\epsilon) \delta_{m,0} (2 a_0 + b_0)  +
\mathcal{O}(\epsilon^2) (2a_m + b_m) =  \epsilon ( \gamma + \delta ) a_m,   \hspace{0.5cm} \label{eq-1} \\
-\epsilon ( b_{m+1} - 2b_m + b_{m-1}) - (1+2\epsilon) \delta_{m,0} (a_0 + 2 b_0)  +
\mathcal{O}(\epsilon^2) (a_m + 2 b_m)
 =  -\left(\frac{3}{2} + \epsilon \gamma - \epsilon \delta \right) b_m. \label{eq-2}
\end{eqnarray}
From the second equation (\ref{eq-2}), we obtain
\begin{eqnarray}  \label{e:b0}
b_0  =  -2(1 - 2\epsilon + 2 \epsilon \gamma - 2 \epsilon \delta + \mathcal{O}(\epsilon^2)) a_0,
\end{eqnarray}
whereas $b_{\pm m} = \mathcal{O}(\epsilon^m) b_0$ for any $m \in
\mathbb{N}$. The first equation (\ref{eq-1}) for any $m \neq 0$
produces the second-order difference equation
$$
-(a_{m+1} - 2a_m + a_{m-1}) + \mathcal{O}(\epsilon) a_m = (\gamma + \delta) a_m, \quad m \in \mathbb{Z} \backslash \{0\},
$$
which admits a unique decaying solution for both $m \to \infty$ and
$m \to -\infty$:
$$
a_m = a_0 e^{-\rho |m|}, \quad m \in \mathbb{Z} \backslash \{0\},
$$
where $\rho$ is a unique root of the transcendent equation
\begin{equation}
\label{transcendent-equation}
\gamma + \delta = 2 - 2 \cosh(\rho), \quad {\rm Re}(\rho) > 0.
\end{equation}
To obtain the value for $\rho$ we close the first equation
(\ref{eq-1}) at $m = 0$:
$$
-2 \epsilon (e^{-\rho} - 1) a_0 - (1+2\epsilon + \mathcal{O}(\epsilon^2)) (2a_0 + b_0) = \epsilon (\gamma + \delta) a_0.
$$
Utilizing (\ref{e:b0}), this equation becomes
$$
-2(e^{-\rho} - 1) - 4(1 - \gamma + \delta) + \mathcal{O}(\epsilon) = \gamma + \delta.
$$
Substituting (\ref{transcendent-equation}) and neglecting the
$\mathcal{O}(\epsilon)$ term, we convert this equation to a
quadratic equation for $z = e^{\rho}$:
$$
3 z^2 + 4(2 \delta - 1) z + 5 = 0,
$$
which admits two possible solutions
$$
z = \frac{2 (1-2 \delta) \pm i \sqrt{15 - 4(1-2\delta)^2}}{3}.
$$
With the help of (\ref{transcendent-equation}), these solutions
produce expressions for $\gamma$ as
$$
\gamma = \frac{14 + 17 \delta \mp i \sqrt{15 - 4(1-2\delta)^2}}{15},
$$
which are complex-valued if $(1 - 2\delta)^2 < \frac{15}{4}$. These
expressions yield the asymptotic approximations
(\ref{result-theorem-5}) and (\ref{result-theorem-4}). From
(\ref{result-theorem-4}), we see that the most unstable eigenvalue
occurs at $\delta=\frac{1}{2}$, which yields the asymptotic approximation
(\ref{e:lammax}).
\end{proof}

\section{Numerical results} \label{sec:num}

In this section, we present numerical results on
transverse-stability eigenvalues of discrete line solitons in one-
and two-dimensional lattices for various values of lattice coupling
parameter $\epsilon$ (with fixed $\mu=1$). These numerical results
are shown to be in good agreement with the analytical results both
qualitatively and quantitatively.

\subsection{Numerical results for the dNLS equation on a two-dimensional lattice}

First we consider discrete line solitons (\ref{s:gamma}) bifurcating
from the $\Gamma$ point in the dNLS equation (\ref{dNLS}). At three
values of $\epsilon$, eigenvalues of the spectral stability problem
(\ref{lin-eigen}) for various transverse wavenumbers $p$ in the
interval $[0, \pi]$ are presented in Fig. \ref{f:fig1} (eigenvalues
for negative $p$ are the same as those for positive $p$). We see
that when $\epsilon=0.1$, a single pair of real eigenvalues exist
for all values of $p$ in $(0, \pi]$, in agreement with Theorem
\ref{theorem-1}. These real eigenvalues closely match the asymptotic
formula (\ref{result-theorem-1}) in Theorem \ref{theorem-1} (middle
left panel). When $\epsilon=1$, this pair of real eigenvalues exist
only in the interval of $0< p < p_0$, where $p_0\approx 2.51$. For
$p>p_0$, these real eigenvalues become purely imaginary. When
$\epsilon=4$, the $p$-interval of real eigenvalues further shrinks
to $(0, p_0)$ with $p_0 \approx 0.91$. Meanwhile, an additional pair
of imaginary discrete eigenvalues appear for all values of $p$ in
$[0, \pi]$. When $\epsilon \to +\infty$, the discrete line soliton
$\psi_m$ approaches the slowly-varying function
(\ref{line-soliton-cont}) with $\mu = 1$, and the interval of real
eigenvalues shrinks to $(0, p_0)$ with $p_0(\epsilon) \to
\sqrt{3/\epsilon}$, according to the elliptic 2D NLS equation
(\ref{elliptic-NLS}) (see \cite{KivPel} and \cite[Section
5.9]{Yang_SIAM}).

\begin{figure}[h]
\center{\includegraphics[width=0.48\textwidth]{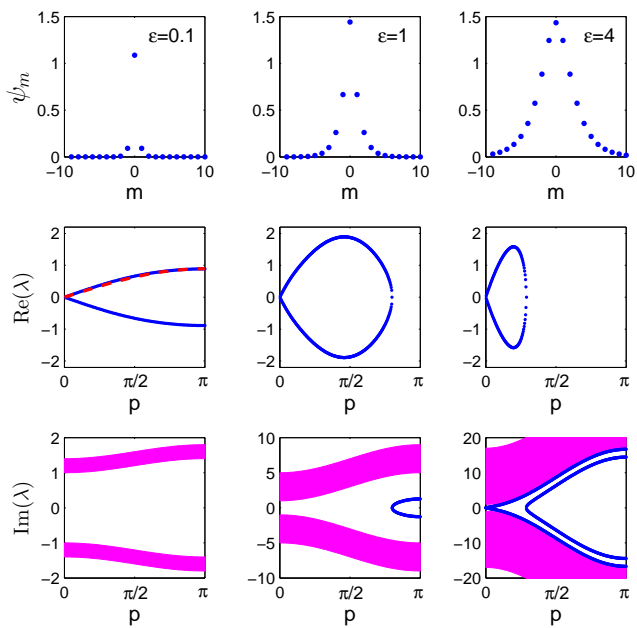}}
\vspace{-0.4cm} \caption{Numerical results for discrete line
solitons (\ref{s:gamma}) bifurcating from the $\Gamma$ point of the
dNLS equation (\ref{dNLS}) on a two-dimensional lattice. Upper row: profiles of discrete line
solitons $\psi_m$; middle row: real parts of eigenvalues $\lambda$
of the spectral stability problem (\ref{lin-eigen})
versus the transverse wavenumber $p$; lower row: imaginary parts of
eigenvalues $\lambda$ versus $p$ (the shaded pink region is the
continuous spectrum). Left column: $\epsilon=0.1$;
middle column: $\epsilon=1$; right column:
$\epsilon=4$. The red dashed line in the middle left panel is the
leading-order analytical approximation (\ref{result-theorem-1}) in
Theorem \ref{theorem-1}. } \label{f:fig1}
\end{figure}

Next we consider discrete line solitons (\ref{s:xpoint}) bifurcating
from the $X$ point in the dNLS equation (\ref{dNLS}). At three values of
$\epsilon$,  eigenvalues of the spectral stability problem
(\ref{lin-eigen2}) for various transverse wavenumbers $p$ in the
interval $[0, \pi]$ are presented in Fig. \ref{f:fig2}. We see that
when $\epsilon=0.01$, a single pair of purely imaginary eigenvalues
exist for all values of $p$ in $(0, \pi]$, in agreement with Theorem
\ref{theorem-2}. These imaginary eigenvalues match the
asymptotic formula (\ref{result-theorem-2}) in Theorem
\ref{theorem-2} (lower left panel). When $\epsilon=0.2$, this pair
of imaginary eigenvalues intersect the continuous spectrum (lower
middle panel). As a consequence, complex eigenvalues appear on the
$p$-interval of $0.97< p < 1.97$ (center panel). When $\epsilon=4$,
additional eigenvalues exist. The eigenvalue curves on the left side
of the $p$-interval (right middle and lower panels) are the
counterparts of similar curves for line solitons in the hyperbolic
2D NLS equation (\ref{hyperbolic-NLS}) (see \cite{DecPel} and
\cite[Section 5.9]{Yang_SIAM}). But the curve of real eigenvalues on
the right side of the $p$-interval (right middle panel) has no
counterpart in the hyperbolic 2D NLS equation
(\ref{hyperbolic-NLS}). These real eigenvalues bifurcate out from
the origin inside the continuous spectrum. As $\epsilon \to
+\infty$, the eigenvalue curves on the left side of the $p$-interval
shrink toward $p=0$ at the asymptotic rate of $\epsilon^{-1/2}$.
Meanwhile, the real-eigenvalue curve on the right side of the
$p$-interval approaches the edge point $p=\pi$, and its width
shrinks at the asymptotic rate of $\epsilon^{-1/2}$.

\begin{figure}[h]
\center{\includegraphics[width=0.48\textwidth]{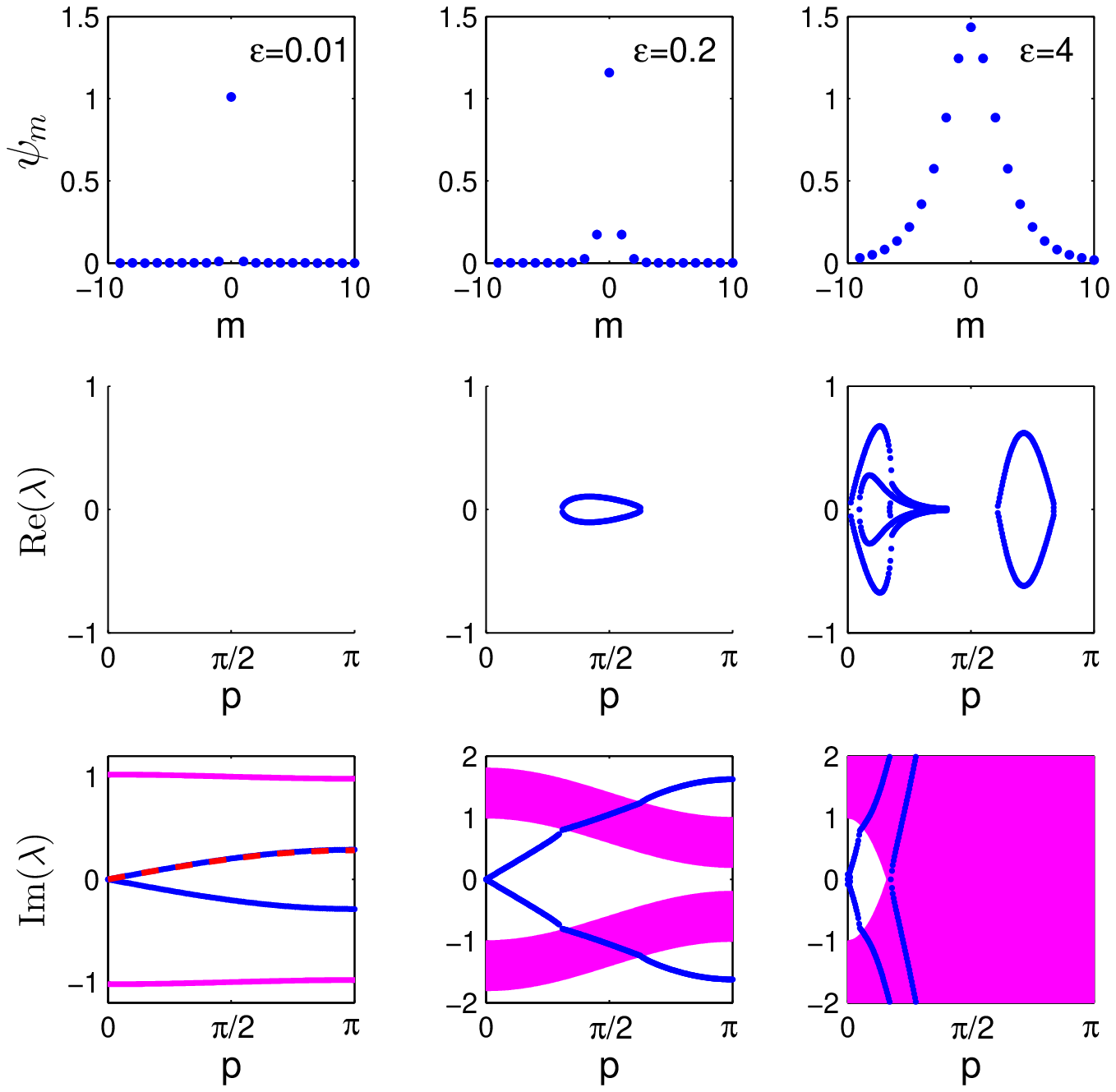}}
\vspace{-0.4cm} \caption{Numerical results for discrete line
solitons (\ref{s:xpoint}) bifurcating from the $X$ point of the
dNLS equation (\ref{dNLS}) on a two-dimensional lattice. Upper row: profiles of discrete line
solitons $\psi_m$; middle row: real parts of eigenvalues $\lambda$
of the spectral stability problem (\ref{lin-eigen2})
versus the transverse wavenumber $p$; lower row: imaginary parts of
eigenvalues $\lambda$ versus $p$ (the shaded pink region is the
continuous spectrum). Left column: $\epsilon=0.01$;
middle column: $\epsilon=0.2$; right column:
$\epsilon=4$. The red dashed line in the lower left panel is the
leading-order analytical approximation (\ref{result-theorem-2}) in
Theorem \ref{theorem-2}.}  \label{f:fig2}
\end{figure}

\subsection{Numerical results for the dNLS equation on a one-dimensional lattice}

Now we consider discrete line solitons (\ref{s:1D}) in the dNLS
equation (\ref{dNLS-diff}) with $\kappa=1$. At three values of
$\epsilon$, eigenvalues of the spectral stability problem
(\ref{lin-eigen3}) for various transverse wavenumbers $p$ are
presented in Fig. \ref{f:fig3}. We see that when $\epsilon=0.1$, a
pair of real eigenvalues exist in the interval $(0,p_0)$, where
$p_0\approx 1.53$. For $p>p_0$, these real eigenvalues become purely
imaginary. This is in agreement with Theorem \ref{theorem-3}.
Quantitatively, these real and imaginary eigenvalues are well
approximated by the leading-order asymptotic formula
(\ref{result-theorem-3}) in Theorem \ref{theorem-3}. At
$\epsilon=2$, we still have the instability band $(0,p_0)$ with $p_0
\approx 1.81$. Meanwhile, two additional branches of purely
imaginary eigenvalues appear over certain $p$-intervals. When
$\epsilon=4$, the instability band $(0,p_0)$ has $p_0 \approx 1.75$,
and one additional branch of purely imaginary eigenvalues exist over
the entire $p$-axis. When $\epsilon \to +\infty$, $p_0(\epsilon) \to
\sqrt{3}$ according to the elliptic 2D NLS equation
(\ref{elliptic-NLS}) \cite{KivPel,Yang_SIAM}.

\begin{figure}[htbp]
\center{\includegraphics[width=0.48\textwidth]{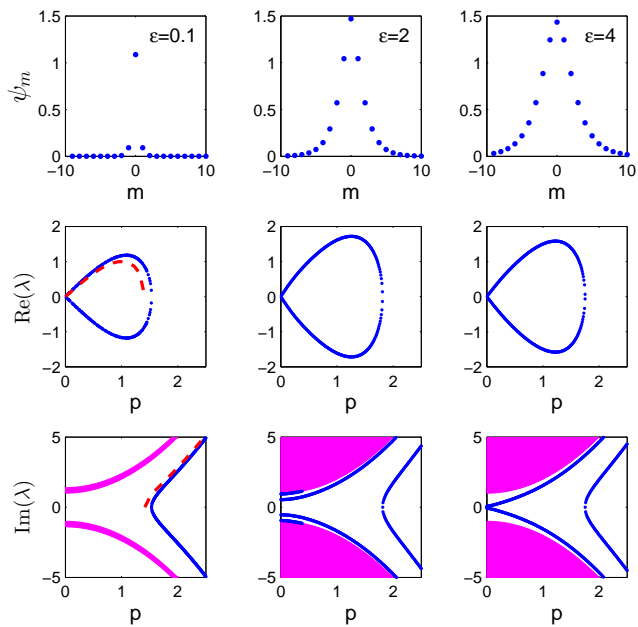}}
\vspace{-0.2cm} \caption{Numerical results for discrete line
solitons (\ref{s:1D}) in the dNLS equation (\ref{dNLS-diff}) on a one-dimensional lattice with
$\kappa=1$. Upper row: profiles of discrete line solitons $\psi_m$;
middle row: real parts of eigenvalues $\lambda$ of the spectral
stability problem (\ref{lin-eigen3}) versus the transverse
wavenumber $p$; lower row: imaginary parts of eigenvalues $\lambda$
versus $p$ (the shaded pink region is the continuous spectrum). Left
column: $\epsilon=0.1$; middle column:
$\epsilon=2$; right column: $\epsilon=4$. The red dashed
lines in the middle and lower left panels are the leading-order
analytical approximations (\ref{result-theorem-3}) in Theorem \ref{theorem-3}.}
\label{f:fig3}
\end{figure}

Next we consider discrete line solitons (\ref{s:1D}) in the dNLS
equation (\ref{dNLS-diff}) with $\kappa=-1$. At the same values of
$\epsilon$, eigenvalues of the spectral stability problem
(\ref{lin-eigen3}) for various transverse wavenumbers $p$ are
presented in Fig. \ref{f:fig4}. We see that when $\epsilon=0.1$, a
pair of imaginary eigenvalues intersect the continuous spectrum
(lower left panel). As a consequence, complex eigenvalues bifurcate
out near $p=1/2$, in agreement with Theorem \ref{theorem-4} (middle
left panel). When $\epsilon=2$, additional eigenvalue bifurcations
occur (middle column). When $\epsilon=4$, eigenvalue curves split
into two parts. The left part is the counterpart of similar curves
for line solitons in the hyperbolic 2D NLS equation
(\ref{hyperbolic-NLS}) \cite{DecPel, Yang_SIAM}, while the right
part is a curve of real eigenvalues at large $p$. Notice that this
eigenvalue structure at $\epsilon=4$ qualitatively resembles that in
Fig. \ref{f:fig2} (right column) for discrete line solitons
bifurcated from the $X$ point in the dNLS equation (\ref{dNLS}). As
$\epsilon \to +\infty$, the left part of this structure
asymptotically approaches eigenvalue curves for line solitons in the
hyperbolic 2D NLS equation (\ref{hyperbolic-NLS}). On the other
hand, the location of the right real-eigenvalue curve moves to $p
\to \infty$ at the asymptotic rate of $\epsilon^{1/2}$, and its
width shrinks at the asymptotic rate of $\epsilon^{-1/2}$.

\begin{figure}[htbp]
\center{\includegraphics[width=0.48\textwidth]{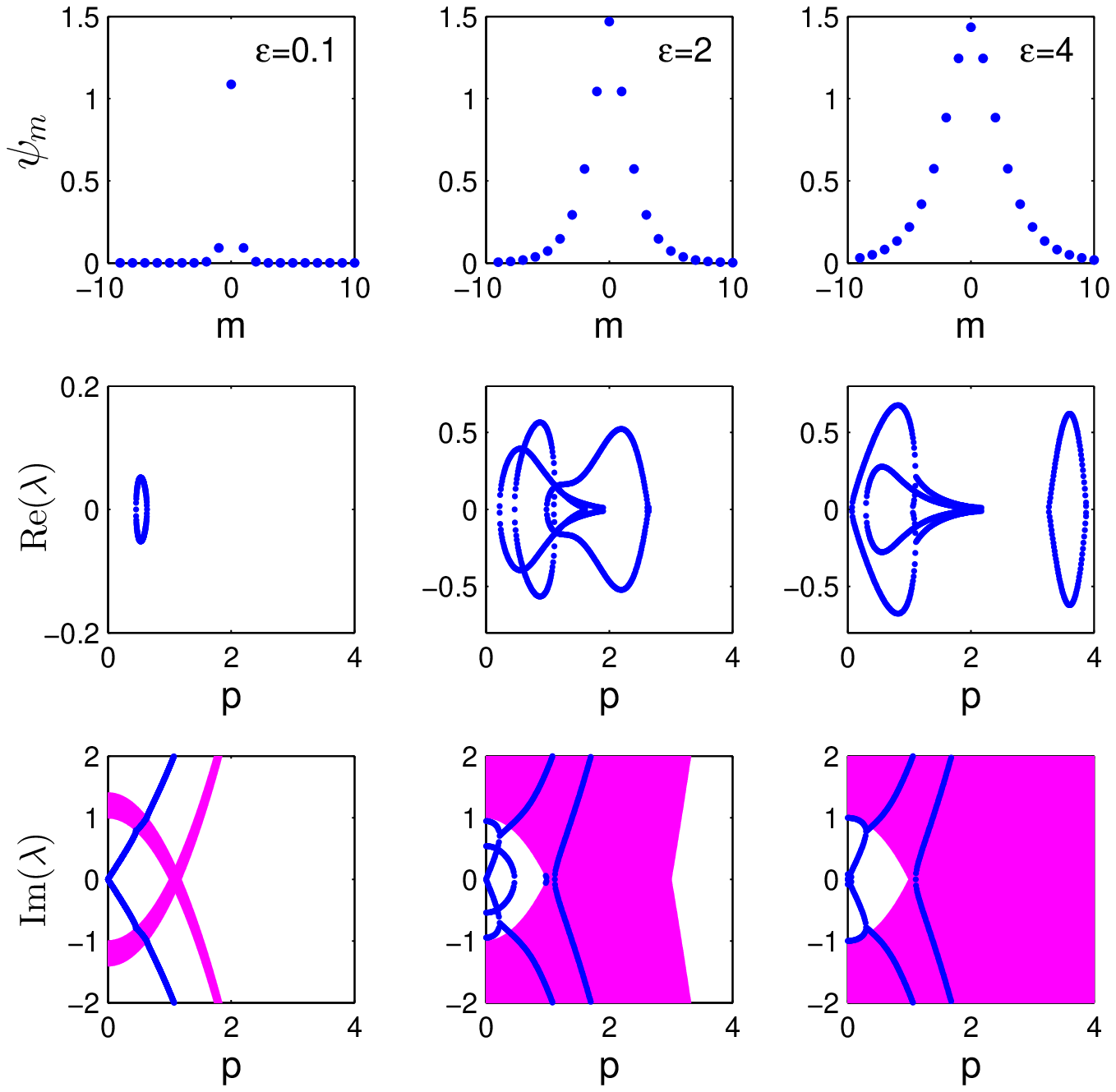}}
\vspace{-0.2cm} \caption{Numerical results for discrete line
solitons (\ref{s:1D}) in the dNLS equation (\ref{dNLS-diff}) on a one-dimensional lattice with
$\kappa=-1$. Upper row: profiles of discrete line solitons $\psi_m$;
middle row: real parts of eigenvalues $\lambda$ of the spectral stability problem (\ref{lin-eigen3})
versus the transverse wavenumber $p$; lower row: imaginary parts of eigenvalues $\lambda$
versus $p$ (the shaded pink region is the continuous spectrum). Left
column: $\epsilon=0.1$; middle column:
$\epsilon=2$; right column: $\epsilon=4$.}
\label{f:fig4}
\end{figure}

Lastly, we quantitatively compare the numerical complex eigenvalues
bifurcating from $p=1/2$ with the analytical formulae for small
$\epsilon$ in Theorem 4. For this purpose, we have numerically
determined the most unstable complex eigenvalue $\lambda_{max}$ and
its $p$-location $p_{max}$ for each $\epsilon$ in the range of
$0<\epsilon<0.3$, and the results are displayed in Fig.
\ref{f:fig5}. For comparison, the leading--order analytical approximations
(\ref{e:lammax}) for $\lambda_{max}$ and $p_{max}$ are also plotted in this
figure. We can see that the analytical and numerical results closely
match each other.

\begin{figure}[htbp]
\center{\includegraphics[width=0.48\textwidth]{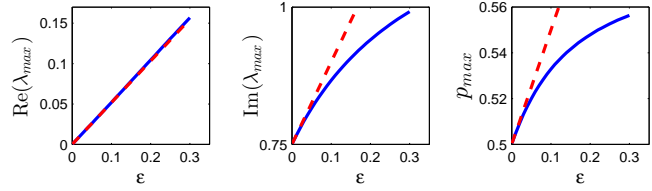}}
\vspace{-0.2cm} \caption{The $\epsilon$-dependence of the most
unstable eigenvalue $\lambda_{max}$ and its corresponding transverse
wavenumber $p_{max}$ for discrete line solitons (\ref{s:1D}) in the
dNLS equation (\ref{dNLS-diff}) with $\kappa=-1$. The red dashed
lines are the leading--order analytical approximations (\ref{e:lammax}) in Theorem 4. } \label{f:fig5}
\end{figure}

\section{Summary and discussion}  \label{sec:summary}

In this article, we have analytically determined the transverse
stability and instability of line solitons in the discrete nonlinear
Schr\"{o}dinger equations on one- and two-dimensional lattices in
the anti-continuum limit. On a two-dimensional lattice, the
fundamental line soliton was proved to be transversely stable
(unstable) when it bifurcates from the $X$ ($\Gamma$) point of the
dispersion surface. On a one-dimensional (stripe) lattice, the
fundamental line soliton was proved to be transversely unstable for
both signs of transverse dispersion. In addition to these
qualitative stability results, we have also derived asymptotic
expressions for unstable eigenvalues and compared them with
numerical results with perfect qualitative and quantitative
agreements.

It is noted that the discrete nonlinear Schr\"{o}dinger equations
are generally used to describe
wave dynamics in the continuous NLS equations with a deep periodic
potential and without inter-band mode coupling. Although the
analytical results in this article nicely explained many of the
numerical results on the transverse stability of line solitons in
the continuous NLS equations \cite{Yang1,Yang2}, they cannot explain
some other notable facts in the continuous models. For instance, our
analytical results for the dNLS equation (\ref{dNLS-diff}) on a
one-dimensional lattice say that all line solitons are transversely
unstable, but the numerical results in \cite{Yang2} showed that in
the continuous model, line solitons near the second Bloch band can
be transversely stable. The reason for this discrepancy
is that line solitons near the second Bloch band contain a strong
coupling between the first and second Bloch bands, which
is neglected in the discrete NLS equation.
How to analytically explain the existence of transversely-stable
line solitons in the continuous NLS equations with a one-dimensional
lattice is still an open issue which merits further study.

{\bf Acknowledgment:} The work of D.P. is supported in part by NSERC.
The work of J.Y. is supported in part by the Air Force Office of Scientific Research
(Grant USAF 9550-12-1-0244) and the National Science Foundation
(Grant DMS-0908167).


\begin{thebibliography}{50}

\bibitem{Aceves_semi_TI}
A.B. Aceves, C. De Angelis, G.G. Luther, A.M. Rubenchik,
``Modulational instability of continuous waves and one-dimensional
temporal solitons in fiber arrays," Opt. Lett. 19, 1186 (1994).

\bibitem{TI_saturable}
N.N. Akhmediev, V.I. Korneev, and R.F. Nabiev, ``Modulation
instability of the ground state of the nonlinear wave equation:
optical machine gun," Opt. Lett. 17, 393--395 (1992).

\bibitem{Anasta} C. Anastassiou, M. Soljacic, M. Segev, E. D. Eugenieva, D. N.
Christodoulides, D. Kip, Z. H. Musslimani, and J. P. Torres,
``Eliminating the transverse instabilities of Kerr solitons", Phys.
Rev. Lett. 85, 4888 (2000).

\bibitem{Bambusi} D. Bambusi and T. Penati, ``Continuous
approximation of breathers in 1D and 2D DNLS lattices", Nonlinearity
{\bf 23}, 143--157 (2010).

\bibitem{Chong} C. Chong, D.E. Pelinovsky, and G. Schneider,
``On the validity of the variational approximation in discrete
nonlinear Schr\"{o}dinger equations", Physica D {\bf 241}, 115--124
(2012).

\bibitem{ChPel} M. Chugunova and D. Pelinovsky, ``Count of unstable
eigenvalues in the generalized eigenvalue problem'', J. Math. Phys.
{\bf 51}, 052901 (2010).

\bibitem{CPV} S. Cuccagna, D. Pelinovsky, and V.Vougalter, ``Spectra of
positive and negative energies in the linearized NLS problem",
Comm.Pure Appl.Math. {\bf 58}, 1--29 (2005).

\bibitem{SHG_TI_neck}
A. De Rossi, S. Trillo, A.V. Buryak, and Y.S. Kivshar,
``Symmetry-breaking instabilities of spatial parametric solitons",
Phys. Rev. E 56, 4959 (1997).

\bibitem{SHG_TI_snake}
A. De Rossi, S. Trillo, A.V. Buryak, and Y.S. Kivshar, ``Snake
instability of one-dimensional parametric spatial solitons," Opt.
Lett. 22, 868-870 (1997).

\bibitem{DecPel} B. Deconinck, D. Pelinovsky, and J.D. Carter,
``Transverse instabilities of deep-water solitary waves", Proc. Roy.
Soc. A {\bf 462}, 2039--2061 (2006).

\bibitem{Gorza}
S. P. Gorza, Ph. Emplit, and M. Haelterman, ``Observation of the
snake instability of a spatially extended temporal bright soliton",
Opt. Lett. 31, 1280 (2006).

\bibitem{Gorza2} S.P. Gorza, B. Deconinck, Ph. Emplit, T. Trogdon, and M.
Haelterman, ``Experimental demonstration of the oscillatory snake
instability of the bright soliton of the (2+1)D hyperbolic nonlinear
Schrodinger equation", Phys. Rev. Lett. 106, 094101 (2011).

\bibitem{Hermann} M. Herrmann, ``Homoclinic standing waves in focusing DNLS equations",
Discrete Contin. Dyn. Syst. {\bf 31}, 737--752 (2011).

\bibitem{Sturm1} R.S. Hilscher, ``Spectral and oscillation theory
for general second order Sturm-Liouville difference equations",
Advances in Difference Equations {\bf 2012}, 82 (2012).

\bibitem{Sturm2} A. Jirari, ``Second-Order Sturm-Liouville Difference Equations and Orthogonal Polynomials",
Memoirs AMS {\bf 542} (1995).

\bibitem{KKS} T. Kapitula, P.G., Kevrekidis, and B. Sandstede,
``Counting eigenvalues via the Krein signature in
infinite-dimensional Hamiltonian systems'', Physica D {\bf 195},
263--282 (2004); Addendum: Physica D {\bf 201}, 199--201 (2005).

\bibitem{KivPel} Yu.S. Kivshar and D.E. Pelinovsky, ``Self-focusing and transverse instabilities of
solitary waves", Phys. Rep. {\bf 331}, 117--195 (2000).

\bibitem{MA94} R.S. MacKay and S. Aubry, ``Proof of existence
of breathers for time-reversible or Hamiltonian networks of weakly
coupled oscillators", Nonlinearity \textbf{7}, 1623--1643 (1994).

\bibitem{Mamaev} A. V.Mamaev, M. Saffman, and A. A. Zozulya,
``Break-up of two-dimensional bright spatial solitons due to
transverse modulation instability", Europhys. Lett. 35, 25 (1996).

\bibitem{Mussli}
Z. H. Musslimani, M. Segev, A. Nepomnyashchy, and Y. S. Kivshar,
``Suppression of transverse instabilities for vector solitons",
Phys. Rev. E 60, R1170 (1999).

\bibitem{MussYang}
Z. H. Musslimani and J. Yang, ``Transverse instability of strongly
coupled dark-bright Manakov vector solitons," Opt. Lett. 26, 1981
(2001).

\bibitem{1Dlattice_TI}
D. Neshev, A.A. Sukhorukov, Y.S. Kivshar, and W. Krolikowski,
``Observation of transverse instabilities in optically induced
lattices," Opt. Lett. 29, 259-261 (2004).

\bibitem{P05}  D.E. Pelinovsky,
``Inertia law for spectral stability of solitary waves in coupled
nonlinear Schr\"{o}dinger equations'', Proc. Roy. Soc. Lond. A, {\bf
461},  783--812 (2005).

\bibitem{Pel-book} D.E. Pelinovsky,
{\em Localization in periodic potentials: from Schr\"{o}dinger operators to the Gross--Pitaevskii
equation} (Cambridge University Press, Cambridge, 2011).

\bibitem{PKF1} D.E. Pelinovsky, P.G. Kevrekidis, and D.J. Frantzeskakis,
``Stability of discrete solitons in nonlinear Schr\"{o}dinger
lattices'', Physica D \textbf{212}, 1--19 (2005).

\bibitem{PelSak} D.E. Pelinovsky and A. Sakovich,
``Internal modes of discrete solitons near the anti-continuum limit
of the dNLS equation", Physica D \textbf{240}, 265--281 (2011).

\bibitem{QinXiao} W.-X. Qin and X. Xiao, ``Homoclinic orbits and localized
solutions in nonlinear Schr\"{o}dinger lattices", Nonlinearity {\bf
20}, 2305--2317 (2007).

\bibitem{VP} V. Vougalter and D. Pelinovsky, ``Eigenvalues of
zero energy in the linearized NLS problem'', J. Math. Phys. {\bf
47}, 062701 (2006).

\bibitem{Yang_SIAM} J. Yang, \emph{Nonlinear Waves in Integrable and Nonintegrable Systems}
(SIAM, Philadelphia, 2010).

\bibitem{Yang1} J. Yang, ``Transversely stable soliton trains in photonic lattice",
Phys. Rev. A {\bf 84}, 033840 (2011).

\bibitem{Yang2} J. Yang, D. Gallardo, A. Miller, and Z. Chen,
``Elimination of transverse instability in stripe solitons by
one-dimensional lattices", Opt. Lett. {\bf 37}, 1571--1573 (2012).

\bibitem{ZakRub}
V.E. Zakharov and A.M. Rubenchik, ``Instability of waveguides and
solitons in nonlinear media", Sov. Phys. JETP 38, 494 (1974).


\end{thebibliography}
\end{document}